\documentclass[journal,12pt,onecolumn]{IEEEtran}

\usepackage{multicol,theorem,amsmath,array}
\usepackage{amssymb}
\usepackage{latexsym}
\usepackage[noadjust]{cite}

\usepackage{algorithm}
\usepackage{algorithmic}

\usepackage{hyperref}

\newtheorem{theorem}{Theorem}
\newtheorem{lemma}[theorem]{Lemma}

\newtheorem{corollary}[theorem]{Corollary}
\newtheorem{example}{Example}
\newtheorem{definition}[theorem]{Definition}
\newenvironment{proof}
{\begin{trivlist}\item[]{{\sc Proof.}}}{\hfill{$\square$}\noindent\end{trivlist}
}

\usepackage{amssymb,amsmath}
\usepackage{mathptmx}       
\usepackage{helvet}         
\usepackage{courier}        
\usepackage{type1cm}        
\usepackage{url}
\usepackage{graphicx}        
\usepackage{multicol}        
\usepackage[bottom]{footmisc}

\usepackage{bm} 
\usepackage[shortlabels]{enumitem}
\usepackage{csquotes}

\newcommand{\F}{\mathbb{F}}
\newcommand{\Z}{\mathbb{Z}}
\newcommand{\Aut}{\mbox{Aut}}
\newcommand{\GL}{\mbox{GL}}

\makeindex             
\begin{document}

\title{Computer classification of linear codes}

\author{Iliya Bouyukliev,
  \thanks{I.~Bouyukliev is with the Institute of Mathematics and Informatics, Bulgarian Academy of Sciences, P.O.
    Box 323, Veliko Tarnovo, BULGARIA. Email:
    iliyab@math.bas.bg}
  Stefka Bouyuklieva,
  \thanks{S.~Bouyuklieva is with the Faculty of Mathematics and Informatics, St. Cyril and St. Methodius University of Veliko Tarnovo, BULGARIA. Email:
    stefka@ts.uni-vt.bg}
    and Sascha Kurz
\thanks{S.~Kurz, is with the
  Department of Mathematics, University
  of Bayreuth, Bayreuth, GERMANY. Email:
  sascha.kurz@uni-bayreuth.de}
}


%



\maketitle

\begin{abstract}
We present algorithms for classification of linear codes over finite fields,
based on canonical augmentation and on lattice point enumeration. We apply these algorithms to obtain classification results over fields with 2, 3 and 4 elements. We
validate a correct implementation of the algorithms with known classification results from the literature, which we partially extend
to larger ranges of parameters.

\end{abstract}

%
\IEEEpeerreviewmaketitle

\section{Introduction}
\label{sec_introduction}
\IEEEPARstart{L}{inear} codes play a central role in coding theory for several reasons. They permit a compact representation via
generator matrices as well as efficient coding and decoding algorithms. Due to their simplicity, i.e., they are just
subspaces of a vector space, they have numerous applications in e.g.\ algebraic geometry and other branches of mathematics
as well as biology or physics. Also multisets of points in the projective
space $\operatorname{PG}(k-1,\mathbb{F}_q)$ of cardinality $n$ correspond to linear $[n,k]_q$ codes, see
e.g.~\cite{dodunekov1998codes}. So, let $q$ be a prime power and $\mathbb{F}_q$ be the field of order $q$. A $q$-ary
linear code of length $n$, dimension $k$, and minimum (Hamming) distance $d$ is called an $[n,k,d]_q$ code.
If we do not want to specify the minimum distance $d$, then we also speak of an $[n,k]_q$ code or of an
$\left[n,k,\left\{w_1,\dots,w_l\right\}\right]_q$ if the non-zero codewords have weights in $\left\{w_1,\dots,w_l\right\}$.
If for the binary case $q=2$ all weights $w_i$ are divisible by $2$, we also speak of an even code.
We can also look at those codes as $k$-dimensional subspaces of the Hamming space $\mathbb{F}_q^n$.
Two linear codes of the same length and dimension are equivalent if one can be obtained from the
other by a sequence of the following transformations: (1) a permutation of the coordinate positions of all codewords; (2) a
multiplication of a coordinate of all codewords with a nonzero element from $\F_q$; (3) a field automorphism.
An $[n,k]_q$ code
can be represented by a generator matrix $G\in\mathbb{F}_q^{k\times n}$ whose row space gives the set of all $q^k$ codewords of the code. In the
remaining part of the paper we always assume that the length $n$ of a given linear code equals its effective length, i.e., for every coordinate
there exists a codeword with a non-zero entry in that coordinate. While a generator matrix gives a compact representation of a linear code it is
far from being unique. Special generator matrices are so-called systematic generator matrices, which contain a $k\times k$ unit matrix in the first $k$
columns. If we apply row operations of the Gaussian elimination algorithm onto a generator matrix we do not change the code itself but just its representation
via a generator matrix. Applying the transformations (1)-(3), mentioned above, we can easily see that each $[n,k]_q$ code admits an equivalent
code with a systematic generator matrix. Already in 1960 Slepian has enumerated binary linear codes for small parameters up to equivalence (or isometry) \cite{slepian1960some}. The general classification problem for $[n,k]_q$
codes has not lost its significance since then, see e.g.\ \cite{betten2006error}. In \cite{jaffe2000optimal} all optimal binary
linear $[n,k,d]_2$ codes up to length $30$ have been completely classified, where in this context optimal means that no $[n-1,k,d]_2$,
$[n+1,k+1,d]_2$, or $[n+1,k,d+1]_2$ code exists. Classification algorithms for linear codes have been presented in \cite{ostergaard2002classifying},
see also \cite[Section 7.3]{kaski2006classification}. A software package \texttt{Q-Extension} is publicly available, see \cite{bouyukliev2007q}
for a description.

The aim of this paper is to present two algorithmic variants for the classification problem for linear codes.The first one is implemented in the program \texttt{Generation} which is a part of
the software package \texttt{QextNewEdition}, and the second one is implemented in the program \texttt{LinCode}. As the implementation of such a software is a delicate issue, we exemplarily verify
several classification results from the literature and partially extend them. Both algorithms are well suited for parallelization. As mentioned in \cite{ostergaard2002classifying},
one motivation for the exhaustive enumeration of linear codes with some specific parameters is that afterwards the resulting codes can
be easily checked for further properties. Exemplarily we do here so for the number of minimal codewords of a linear code, see Subsection~\ref{subsec_applications}.

The remaining part of the paper is organized as follows. In Section \ref{sec:algorithms} we present two versions of our algorithm for canonical augmentation - extension of a generator matrix column by column or row by row. The details and the theoretical foundation of the other algorithm is given in Section~\ref{sec_extending}. Numerical enumeration and classification results for linear codes are listed in Section~\ref{sec_results}.
Finally, we draw a brief conclusion in Section~\ref{sec_conclusion}.

\section{Classification of linear codes using canonical augmentation}
\label{sec:algorithms}

The concept of canonical augmentation is introduced in \cite{kaski2006classification} and \cite{McKay}. The main idea is to construct only nonequivalent objects (in our case - inequivalent linear codes) and in this way to have a classification of these objects. The construction is recursive, it consists of steps in which the nonequivalent objects are obtained from smaller objects by expanding in a special way.  The canonical augmentation uses a canonical form to check the so called "parent test" and considers only objects that have passed the test.

The technique of canonical augmentation has been used for classification of special types of codes and related combinatorial objects in \cite{bouyukliev2006,BB38,vanEupenLizonek1997,royle1998orderly}, etc.
The algorithms in the pointed works construct objects with the needed parameters recursively starting from the empty set. In this way, to classify all linear $[n,k]_q$ codes column by column, the codes of lengths $1,2,\dots,n$ and dimensions $\le k$ are also constructed in the generation process. One of the important differences of the algorithm presented here is that we consider only codes with dimension $k$.

\subsection{Preliminaries}
\label{sec:preliminaries}

To construct all inequivalent codes with given parameters means to have one representative of each equivalence class. To do this, we use the concept for a canonical representative, selected
on the base of some specific conditions. This canonical
representative is intended to make easily a distinction between
the equivalence classes.

Let $G$ be a group acting on a set $\Omega$. This action defines an equivalence relation such that the equivalence classes are the $G$-orbits in $\Omega$. We wish to find precisely one representative of each $G$-orbit and therefore we use a so-called canonical representative map.

\begin{definition} {\rm\cite{kaski2006classification}} A canonical representative map for the
action of the group $G$ on the set $\Omega$ is a function
$\rho:\Omega\rightarrow\Omega$ that satisfies the following
two properties:
\begin{enumerate}
\item for all $X\in\Omega$ it holds that $\rho(X)\cong X$,
\item for all $X, Y\in\Omega$ it holds that $X\cong Y$ implies
$\rho(X) = \rho(Y)$.
\end{enumerate}
\end{definition}

For $X\in\Omega$, $\rho(X)$ is the canonical
form of $X$ with respect to $\rho$. Analogously, $X$ is in
canonical form if $\rho(X)=X$. The configuration $\rho(X)$ is the canonical
representative of its equivalence class with respect to $\rho$. We can take for a canonical representative of one equivalence
class a code which is more convenient for our purposes.



We take $\Omega$ to be the set of all linear $[n,k,\ge d]_q$ codes with dual distance at least $d^\perp$, and $G$ be the semidirect product $(\F_q^*\wr S_n)\rtimes_{\theta} \Aut(\F_q)$ where $\theta:\Aut(\F_q)\to \Aut(\F_q^*\wr S_n)$ is a homomorphism such that $\theta_{\alpha}((z,h))=(\alpha(z),h)$ for all $\alpha\in\Aut(\F_q)$ and $(z,h)\in \F_q^*\wr S_n$ (for more details see \cite{kaski2006classification}). The elements of $G$ fix the minimum and the dual distance of the codes. Using that $\F_q^*\wr S_n\cong Mon_n(\F_q)$ where $Mon_n(\F_q)$ is the group of the monomial $n\times n$ matrices over $\F_q$, we can consider the elements of $G$ as pairs $(M,\alpha)$, $M\in Mon_n(\F_q)$, $\alpha\in\Aut(\F_q)$. An automorphism of the linear code $C$ is a pair $(M,\alpha)\in Mon_n(\F_q)\rtimes \Aut(\F_q)$ such that $vM\alpha\in C$ for any codeword $v\in C$.
The set of all automorphisms of the code $C$ forms the automorphism group
of $C$, denoted by $\Aut(C)$. For linear codes over a prime field the transformations are of types (1) and (2) and a sequence of such transformations can be represented by a monomial matrix over the considered field. For binary codes, the transformations (2) and (3) are trivial and therefore $\Aut(C)$ is a subgroup of the symmetric group $S_n$.

We use one more group action. The automorphism group of the code $C$ acts on the set of coordinate positions and partitions them into orbits. The canonical representative map $\rho$ induces an ordering of these orbits.
 The all-zero coordinates, if there are any, form an orbit which we denote by $O_a$. If the code contains codewords of weight 1 then their supports form one orbit, say $O_b$.  The orbits for the canonical representative code $\rho(C)$ are ordered in the following way:
 $O^{(\rho)}_1$ contains the smallest integer in the set $\{1,2,\ldots,n\}\setminus (O^{(\rho)}_a\cup O^{(\rho)}_b)$, $O^{(\rho)}_2$ contains the smallest integer which is not in the set $O^{(\rho)}_a\cup O^{(\rho)}_b\cup O^{(\rho)}_1$, etc.
      If $\phi$ maps the code $C$ into its canonical form $\rho(C)$ then the permutational part $\pi_\phi$ of $\phi$
maps the orbits of $C$ into the orbits of $\rho(C)$. Obviously, $\phi(O_a)=O^{(\rho)}_a$ and $\phi(O_b)=O^{(\rho)}_b$. If $\pi_\phi(O_{i_s})=O^{(\rho)}_s$ then $O_{i_1}\prec O_{i_2}\prec \cdots\prec O_{i_m}$. We call the first orbit $O_{i_1}$ special and denote it by $\sigma(C)$.
If $\{1,2,\ldots,n\}=O_a\cup O_b$ then the code contains only codewords with weights $0$ and $1$, and in this case we do not define a special orbit.

\begin{example}
If we order the codewords in a code
lexicographically and then compare the codes according to a
lexicographical ordering of the vectors obtained by concatenation of the ordered nonzero codewords, we can take the smallest code in any equivalence class as a canonical representative. This type of canonical map is very easy to define but computationally expensive to implement.
Consider the binary code $C$ generated by the matrix
$G_C=\displaystyle\left(\begin{array}{cccc}1&0&1&1\\ 0&1&0&1\end{array}\right)$
in details.
The automorphism group of $C$ is $\Aut(C)=\{ id, (13),(24),(13)(24)\}$.
If $\Omega_C$ is the equivalence class of $C$ then $\Omega_C=\{ C_1,\ldots,C_6\}$,
$C_i=\{0,v^{(i)}_1, v^{(i)}_2, v^{(i)}_3\}$, $v^{(i)}_1\prec v^{(i)}_2\prec v^{(i)}_3$. We order the codes in $\Omega_C$ in the following way:
$$C_i\prec C_j\iff (v^{(i)}_1,v^{(i)}_2,v^{(i)}_3)\prec (v^{(j)}_1,v^{(j)}_2,v^{(j)}_3).$$
Therefore,
$C=\{0,0101,1011,1110\}\succ C_1=\{0,0011,1101,1110\}$.  Hence the code $C_1$ is the canonical form of $C$, $C_1=\rho(C)$.
The coordinates of $C_1$ are partitioned into two orbits under the action of its automorphism group, namely $O_1=\{1,2\}\prec O_2=\{3,4\}$. For the code $C$ the special orbit is $\sigma(C)=\{1,3\}$.
\end{example}

To find a canonical form of a code, we use the algorithm described in \cite{bouyukliev2007code}. Similarly to the McKay's program \texttt{nauty} \cite{mckay1990nauty}, this algorithm gives in addition to canonical form, also generating elements of the automorphism group of the considered code. Note that if the coordinates are previously partitioned according to suitable invariants, the algorithm works much faster.


Using the concept of canonical augmentation, we have developed an algorithm in two variants.

\subsection{Algorithm 1}
\label{Algorithm_1}

The first algorithm is a canonical augmentation column by column. We are looking for all inequivalent linear codes with length $n$, dimension $k$, minimum distance $\ge d$ and dual distance at least $d^\perp\ge 2$.
Without loss of generality we can consider the generator matrices in the form $(I_k\vert A)$ where $A$ is a $k\times (n-k)$ matrix. To obtain the codes we use a recursive construction starting with the identity matrix $I_k$ which generates the trivial $[k,k,1]_q$ code. In the $i$-th step we add a column to the considered generator matrices of the obtained $[k+i-1,k]_q$ codes but we take only those columns which gives codes of length $k+i$ with minimum distance $\ge d_i=d-(n-k)+i$ and dual distance at least $d^\perp$. A strategy for effective generation of these vectors (columns) is described in \cite{IliyaMaya}. Since $d\le n-k+1$, the minimum distance in the beginning is $\le 1$ (it is equal to 1 as we begin with the trivial code). The codes obtained from a code $C$ in this way form the set $Ch(C)$ and they are called the children of $C$.
We say that the code $\overline{C}\in Ch(C)$ passes the parent test, if the added coordinate belongs to the special orbit $\sigma(\overline{C})$. Moreover, we define an action of the automorphism group $\Aut(C)$ on the set of all vectors in $\F_q^k$ and take only one representative from each orbit. By $Ch^*(C)$ we denote a subset of $Ch(C)$ consisting of the codes constructed by $C$ and the taken representatives.

\begin{algorithm}[ht]
\caption{Canonical augmentation column by column}
\label{Alg1}
\begin{algorithmic}
\REQUIRE The trivial $[k,k,1]_q$ code $C_k$
\ENSURE A set $U_{n}$ of linear $[n,k,\ge d]_q$ codes with dual distance $\ge d^\perp$
\STATE \textbf{Function Augmentation}($A$: linear code of dimension $k$);
\IF {the length of $A$ is equal to $n$ }
\STATE $U_n:= U_n\cup \{A\}$;
\ELSE
\FOR {all codes $B\in Ch^*(A)$}
\IF {$B$ passes the parent test}
\STATE Augmentation($B$);
\ENDIF
\ENDFOR
\ENDIF
\STATE \textbf{Function Main};
\STATE $U_{n}=\emptyset$
\STATE Augmentation($C_k$);
\end{algorithmic}
\end{algorithm}

Using some lemmas we will prove the following theorem:
\begin{theorem}\label{thm:main1}
The set $U_n$ obtained by Algorithm \ref{Alg1}
consists of all inequivalent $[n,k,\ge d]_q$ codes with dual distance at least $d^\perp$.
\end{theorem}

The main idea is to prove that Algorithm \ref{Alg1} gives a tree of codes with root the trivial code $C_k$. The codes obtained in level $i$ represents all inequivalent $[k+i,k]_q$ codes with minimum distance at least $d_i$ and dual distance at least $d^\perp$. Denote the set of these codes by $U_{k+i}$. We have to prove that all constructed codes in $U_{k+i}$ are inequivalent, and that any $[k+i,k]_q$ code with needed minimum and dual distance is equivalent to a code in this set.


The first lemma proves that the equivalence test for codes that pass the parent test and are obtained from non-equivalent parent codes is not necessary.

\begin{lemma}\label{Lemma:parent}
If $B_1$ and $B_2$ are two equivalent linear $[n,k,d]$ codes
which pass the parent test, their parent codes are also
equivalent.
\end{lemma}

\begin{proof}
Let $B=\rho(B_1)=\rho(B_2)$ be the canonical representative of the equivalence class of the considered codes.
Since both codes pass the
parent test, then the added column is in the special orbit of both codes, or $n\in\sigma(B_i)$, $i=1,2$. This means that there is a map $\psi$ that maps $B_1$ to $B_2$ and the permutational part of $\psi$ fixes $n$-th coordinate.
Hence $\psi=(M,\alpha)$, $M=\left(\begin{array}{cc}M_1&0\\ 0&\lambda\\ \end{array}\right)\in Mon_n(\F_q)$, $\lambda\in\F_q^*$, $\alpha\in \Aut(\F_q)$, and $(M_1,\alpha)$ maps
the parent code of $B_1$ to the parent code of $B_2$. Hence both parent codes are equivalent.
\end{proof}

\begin{lemma}\label{Lemma:equ-parents}
Let $A_1$ and $A_2$ be two equivalent linear codes of length
$r$ and dimension $k$. Then for any child code $B_1$ of $A_1$ which passes the
parent test, there is a child code $B_2$ of $A_2$, equivalent
to $B_1$, such that $B_2$ also passes the parent test.
\end{lemma}

\begin{proof}  Let $G_1$ be a generator matrix of
$A_1$ in systematic form, and $A_2=\psi(A_1)$, $\psi=(M,\alpha)$, $M\in Mon_{r}(\F_q)$, $\alpha\in\Aut(\F_q)$. Let $B_1$ be the code generated by $(G_1\vert a^T)$, $a\in\F_q^k$, and
$B_2$ be the code generated by the matrix $G_2=\psi(G_1)$ and the
vector $b^T=(a^\alpha)^T$, where $a^\alpha$ is obtained from $a$ by applying the field automorphism $\alpha$ to all coordinates. Extend the map $\psi$ to $\widehat{\psi}=(\left(\begin{array}{cc}M&0\\ 0&1\\ \end{array}\right),\alpha)\in Mon_{r+1}(\F_q)\rtimes \Aut(\F_q)$ so $\widehat{\psi}(v,v_{r+1})=(vM,v_{r+1})^\alpha$.  Then $$(G_1\vert a^T)\left(\begin{array}{cc}M&0\\ 0&1\\ \end{array}\right)\alpha=
(G_1M\vert a^T)^\alpha=(G_2\vert b^T)$$
and $B_2=\widehat{\psi}(B_1)$. Hence the codes $B_1$ and $B_2$ are equivalent and so they have the same
canonical representative $B=\rho(B_1)=\rho(B_2)$.

The code $B_1$
passes the parent test and therefore the added column is in the special orbit. Since
$\phi_1\widehat{\psi}^{-1}(B_2)=\phi_1(B_1)=\rho(B_1)=\rho(B_2)$,
$\phi_2=\phi_1\widehat{\psi}^{-1}$ maps $B_2$ to its canonical form $B$. Since $\phi_2$ acts on the added coordinate in the same way as $\phi_1$, this coordinate is in the special orbit
and therefore the code
$B_2$ also passes the parent test.
\end{proof}

To see what happens with the children of the same code $C$, we have to consider the automorphism group of $C$ and the group $G=Mon_n(\F_q)\rtimes \Aut(\F_q)$ which acts on all linear $[n,k]_q$ codes (for more details on this group see \cite{HuffmanPless}).
A monomial matrix $M$ can be written either in the form $DP$ or the
form $PD_1$, where $D$ and $D_1$ are diagonal matrices and $P$ is a permutation matrix, $D_1=P^{-1}DP$. The multiplication in the group $Mon_n(\F_q)\rtimes \Aut(\F_q)$ is defined by $(D_1P_1\alpha_1)(D_2P_2\alpha_2)=(D_1(P_1D_2^{\alpha_1^{-1}}P_1^{-1})P_1P_2\alpha_1\alpha_2)$, where $B^{\alpha}$ denotes the matrix obtained by $B$ after the action of the field automorphism $\alpha$ on its elements. Obviously, $(AB)^{\alpha}=A^\alpha B^\alpha$ and $P^\alpha=P$ for any permutation matrix $P$.
Let see now what happens if we take different vectors $a, b\in\F_q^k$ and use them in the construction extending the
same linear $[n,k]_q$ code $C$ with a generator matrix
$G_C$. We define an action of the automorphism group $\Aut(C)$ of the code $C$ on the set of all vectors in $\F_q^k$. 
To any automorphism $\phi\in \Aut(C)$ we can correspond an invertible matrix $A_\phi\in
\GL(k,q)$ such that $G'=G_C\phi=A_\phi G_C$, since $G'$ is another
generator matrix of $C$. Using this connection, we obtain a homomorphism $f
\ : \ \Aut(C) \longrightarrow  \GL(k,q)\rtimes \Aut(\F_q)$, $f(M,\alpha)=(A_\phi,\alpha)$.
We have
\begin{align*}
  G_C\phi_1\phi_2 & =(A_{\phi_1}G_C)\phi_2=(A_{\phi_1}G_C)M_2\alpha_2
    =(A_{\phi_1}G_C)^{\alpha_2}M_2^{\alpha_2}\\
   &=A_{\phi_1}^{\alpha_2}G_C^{\alpha_2}M_2^{\alpha_2}=A_{\phi_1}^{\alpha_2}A_{\phi_2}G_C.
\end{align*}
Hence $A_{\phi_1\phi_2}=A_{\phi_1}^{\alpha_2}A_{\phi_2}$ and so $f(\phi_1\phi_2)=f(\phi_1)f(\phi_2)$, when the operation in the group $\GL(k,q)\rtimes \Aut(\F_q)$ is $(A,\alpha)\circ (B,\beta)=(A^\beta B,\alpha\beta)$.
Consider the action
of $Im (f)$ on the set $\F_q^{k}$ defined by $(A,\alpha)(x)=(Ax^T)^{\alpha^{-1}}$ for
every $x\in \F_q^{k}$.

\begin{lemma}\label{lemma:ab}
Let $a,b\in\F^k_q$. Suppose that $a^T$ and
$b^T$ belong to the same $Im(f)$-orbit, where $a^T$ denotes the
transpose of $a$. Then the $[n + 1, k]_q$ codes with generator matrices
$(G_C \ a^T)$ and $(G_C \ b^T)$ are equivalent and if one of them passes the parent test, the other also passes the test. Moreover, if the codes with generator matrices
$(G_C \ a^T)$ and $(G_C \ b^T)$ are equivalent and pass the parent test, the vectors $a^T$ and
$b^T$ belong to the same $Im(f)$-orbit.
\end{lemma}

\begin{proof}
Let the matrices $(G_C\vert a^T)$ and $(G_C\vert b^T)$  generate the codes $C_1$ and $C_2$, respectively, and
$b^T=(A_\phi a^T)^{\alpha^{-1}}$, where $\phi=(M,\alpha)\in\Aut(C)$. Then
$$\widehat{\phi}(G_C\vert b^T)=(G_CM\vert b^T)^{\alpha}=((G_CM)^{\alpha}\vert (b^T)^{\alpha})=
(A_\phi G \ A_\phi a^T)=A_\phi(G \ a^T),$$
where $\widehat{\phi}=(\left(\begin{array}{cc}M&0\\ 0&1\\ \end{array}\right),\alpha)\in Mon_{n+1}(\F_q)\rtimes \Aut(\F_q)$.
Since $A_\phi(G \ a^T)$ is another generator
matrix of the code $C_1$, both codes are equivalent. Moreover, the permutational part of $\widehat{\phi}$ fixes the last coordinate position, hence if $n+1$ is in the special orbit of $C_1$, it is in the special orbit of $C_2$ and so both codes pass (or don't pass) the parent test.

Conversely, let $C_1\cong C_2$ and both codes pass the parent test. It turns out that there is a
map $\psi=(M_{\psi},\beta)\in G$ such that $\psi(C_1)=C_2$ and $\pi_\psi(n+1)=n+1$ where $\pi_\psi$ is the permutational part of $\psi$. Hence $M_\psi=\left(\begin{array}{cc}M_1&0\\ 0&\mu\\ \end{array}\right)$ and
$$(G_C\vert a^T)M_\psi\beta=(G_C M_1\vert \mu a^T)\beta=(G_C M_1\beta\vert (\mu a^T)^\beta)=A(G_C\vert b^T).$$
It follows that $G_C M_1\beta=AG_C$ which means that $(M_1,\beta)\in\Aut(C)$, and $(\mu a^T)^\beta=Ab^T$, so $a^T=((\mu^{-1})^\beta Ab^T)^{\beta^{-1}}$. Since
$$G(\mu^{-1}M_1,\beta)=(\mu^{-1}GM_1)\beta=(\mu^{-1})^\beta (GM_1)^\beta=(\mu^{-1})^\beta AG,$$
we have $((\mu^{-1})^\beta A,\beta)=f(\mu^{-1}M_1,\beta)$.
Hence $(\mu^{-1}M_1,\beta)\in\Aut(C)$ and $a^T$ and $b^T$ belong to the same orbit under the defined action.
\end{proof}

\emph{Proof of Theorem \ref{thm:main1}:}

The algorithm starts with the trivial $[k,k,1]_q$ code $C_k=\F_q^k$. In this case $\Aut(C_k)=Mon_{k}(\F_q)\rtimes \Aut(\F_q)$ and the group partitions the set $\F_q^k$ into $k+1$ orbits as two vectors are in the same orbit iff they have the same weight. We take exactly one representative of each orbit (instead the zero vector) and extend $I_k$ with these column-vectors. If $d_1=2$, we take only the obtained $[k+1,k,2]_q$ code, otherwise we take all constructed codes and put them in the set $ch^*(C)$. All obtained codes pass the parent test.

Suppose that $U_{k+i}$ contains inequivalent $[k+i,k,\ge d_i]_q$ codes with dual distance $\ge d^\perp$, $d_i=d-n+k+i$, and any code with these parameters is equivalent to a code in $U_{k+i}$.
We will show that the set $U_{k+i+1}$ consists only of inequivalent codes, and
any linear $[k+i+1,k,\ge d_{i+1}]_q$ code is equivalent to a code in
the set $U_{k+i+1}$.

Suppose that the codes $B_1, B_2\in U_{k+i+1}$ are equivalent. Since
these two codes have passed the parent test, their parent codes are also equivalent according to Lemma \ref{Lemma:parent}.
These parent codes are linear codes from the set $U_{k+i}$
which consists only in inequivalent codes. The only option for both codes is to have the same parent. But as we take only one vector of each orbit under the considered group action, we obtain only inequivalent children from one parent code (Lemma \ref{lemma:ab}). Hence $B_1$ and $B_2$ cannot be equivalent.

Take now a linear $[k+i+1,k,\ge d_{i+1}]_q$ code $C$ with a
canonical representative $B$. If $\sigma(C)$ is the special orbit, we can reorder the coordinates of $C$ such that one of the coordinates in $\sigma(C)$ to be the last one. So we obtain a code $C_1$ that is permutational equivalent to $C$ and passes the parent test. Removing this coordinate, we obtain a parent code $C_P$ of $C_1$. Since $U_{k+i}$ consists of all
inequivalent $[k+i,k,\ge d_i]_q$ codes with dual distance $\ge d^\perp$, the parent
code $C_P$ is equivalent to a code $A\in U_{k+i}$. According to Lemma
\ref{Lemma:equ-parents}, to any child code of $C_P$ that passes the parent test, there is a child code of $A$ that also passes the test. So there is a child code $C_A$ of $A$ that passes the test, so $C_A\in U_{k+i+1}$, and $C_A$ is equivalent to $C$.
In this way we find a code in $U_{k+i+1}$ which is equivalent to $C$.

Hence in the last step we obtain all inequivalent $[n,k,\ge d]_q$ codes with the needed dual distance.

\medskip
Our goal is to get all linear $[n,k]_q$ codes with given dual distance starting from the $k\times k$ identity matrix. We can also start with all already constructed $[n'<n,k]_q$ codes to get all $[n,k]_q$ codes with the needed properties. Similar algorithms are developed in \cite{vanEupenLizonek1997,royle1998orderly} but these algorithms start from the empty set and generate all inequivalent codes of length $\le n$ and dimensions $1,2,\dots,k$.

\subsection{Algorithm 2}

The second algorithm is a canonical augmentation row by row. We start from the empty set (or set of already given codes with parameters $[n-i,k-i,d]_q$, $1\le i\le k$) and aim to construct all $[n,k,\ge d]_q\ge d^\perp$ codes. In any step we add one row and one column to the considered generator matrix. In the $i$-th step we extend the $[n-k+i-1,i-1,\ge d]_q$ codes to $[n-k+i,i,\ge d]_q$ codes.

We consider generator matrices in the form $(A\vert I_k)$. If $C$ is a linear $[n-k+s,s,\ge d]_q$ code with a generator matrix $(A\vert I_{s})$, we extend the matrix to $\left(\begin{array}{c|c|l} A& I_{s}&0^T\\
\hline a&0\ldots 0&1\\ \end{array}\right)=
\left(\left.\begin{array}{c} A\\ a\\ \end{array} \right| I_{s+1}\right)$, where $a\in\F_{n-k}$. If our aim is to construct codes with dual distance $d^\perp_k\ge d^\perp$, in the $s$-th step we need codes with dual distance $d_s^\perp\ge d^\perp-(k-s)$. The obtained $[n-k+s+1,s+1,\ge d]_q$ codes with dual distance $\ge d^\perp-(k-s)$ are the children of $C$ and the set of all such codes is denoted by $Ch(C)$. The parent test for these codes is the same as in Algorithm \ref{Algorithm_1}. We take a canonical representative for the dual code of $C$ such that $\rho(C^\perp)=\rho(C)^\perp$. The orbits of $C$ are ordered in the same way as the orbits of $C^\perp$ and the special orbit for both codes is the same. The only difference is that if $C$ is a code with zero coordinates then the orbit consisting of these coordinates coincides with the orbit of $C^\perp$ consisting of the supports of the codewords with weight $1$. As in the previous algorithm, we define a group action but now on the vectors in $\F_q^{n-k}$ and take one representative from each orbit for the construction. The corresponding set of codes is denoted by $Ch^*(C)$. Lemma \ref{Lemma:parent} and Lemma \ref{Lemma:equ-parents} hold in this case, too.

If $(A\vert I_k)$ is a generator matrix of $C$ then $(I_{n-k}\vert -A^T)$ generates $C^\perp$. So in the extension in the $s$-th step the vector $-a^T$ expands the considered generator matrix of $C^\perp$ to give a generator matrix of the extended code $\overline{C^\perp}\in Ch(C^\perp)$.
Moreover, $\Aut(C^\perp)=\{ (D^{-1}P,\alpha)\vert (DP,\alpha)\in\Aut(C)\}$. Therefore, for the action of $\Aut(C)$ on the vectors in $\F_q^{n-k}$, we use the elements of $\Aut(C^\perp)$.
If $\phi=(DP,\alpha)\in \Aut(C)$ then $\phi'=(D^{-1}P,\alpha)\in \Aut(C^\perp)$ and so we have an invertible matrix $B_\phi\in
\GL(n-k,q)$ such that $G'=(I_k\vert -A^T)\phi'=B_\phi (I_k\vert -A^T)$, since $G'$ is another
generator matrix of $C^\perp$. In this way we obtain a homomorphism $f'
\ : \ \Aut(C) \longrightarrow  \GL(n-k,q)\rtimes \Aut(\F_q)$, $f(DP,\alpha)=(B_\phi,\alpha)$.
Then we consider the action
of $Im (f')$ on the set $\F_q^{n-k}$ defined by $(B,\alpha)(x)=(Bx^T)^{\alpha^{-1}}$ for
every $x\in \F_q^{n-k}$. This action is similar to the action defined in Subsection \ref{Algorithm_1}. The proof of the following lemma for an $[n,k]$ code $C$ with a generator matrix $(A\vert I_k)$ is similar to the proof of Lemma \ref{lemma:ab}.

\begin{lemma}\label{lemma:ab2}
Let $a,b\in\F^{n-k}_q$. Suppose that $a$ and
$b$ belong to the same $Im(f')$-orbit. Then the $[n + 1, k+1]_q$ codes with generator matrices
$\left(\left.\begin{array}{c} A\\ a\\ \end{array} \right| I_{k+1}\right)$ and $\left(\left.\begin{array}{c} A\\ b\\ \end{array} \right| I_{k+1}\right)$ are equivalent and if one of them passes the parent test, the other also passes the test. Moreover, if the codes with generator matrices
$\left(\left.\begin{array}{c} A\\ a\\ \end{array} \right| I_{k+1}\right)$ and $\left(\left.\begin{array}{c} A\\ b\\ \end{array} \right| I_{k+1}\right)$ are equivalent and pass the parent test, the vectors $a$ and
$b$ belong to the same $Im(f')$-orbit.
\end{lemma}

The proof that Algorithm 2 gives the set $U_n$ of all inequivalent $[n,k,\ge d]_q$ codes with dual distance $\ge d^\perp$ is similar to the proof of Theorem \ref{thm:main1}, therefore we skip it.


%

\subsection{Some details}

The parent test is an expensive part of the algorithms. That's way we use invariants to
take information about the orbits $\{O_1,\ldots,O_m\}$ after the action of $\Aut(C)$ on the set of coordinate positions.
An invariant of the coordinates of $C$ is a function $f:
N\to\Z$ such that if $i$ and $j$ are in the same orbit with
respect to $\Aut(C)$ then $f(i)=f(j)$, where $N=\{1,2,\dots,n\}$ is the set of the coordinate positions.
The code $C$ and the invariant $f$ define a partition $\pi= \{
N_1,N_2,\dots,N_l\}$ of the coordinate set $N$,  such that
$N_i\cap N_j=\emptyset$ for $i\not =j$, $N=N_1\cup
N_2\cup\dots\cup N_l$, and two coordinates $i,j$ are in the same
subset of $N \iff  f(i)= f(j)$. So the subsets $N_i$ are unions of
orbits, therefore we call them pseudo-orbits. We can use the fact
that if we take two coordinates from two different subsets, for
example $s\in N_i$ and $t\in N_j$, $N_i\cap N_j=\emptyset$, they
belong to different orbits under the action of $\Aut(C)$ on the
coordinate set $N$. Moreover, using an invariant $f$, we can
define a new canonical representative and a new special orbit of $C$ in the following way.
If $f_i=f(j_i)$ for $j_i\in N_i$, $i=1,2,\dots,l$, we can order the pseudo-orbits in respect to the integers $f_i$. We take for a canonical representative a code for which $f_1<f_2<\cdots <f_l$. Moreover, we order the orbits in one pseudo-orbit as it is described in Section \ref{sec:preliminaries}. So the orbits in the canonical representative are ordered according this new ordering. The special orbit for a code $C$ is defined in the same way as in Section \ref{sec:preliminaries} (only the canonical map and the canonical representative may be different).

In the step "\textit{if $B$ passes the parent test}", using a
given generator matrix of the code $B$ we have to calculate
invariants, and in some cases also canonical form and the
automorphism group $\Aut(B)$. Finding a canonical form and the
automorphism group is necessary when the used invariants are not
enough to prove whether the code $B$ pass or not the parent test.
If the code $B$ passes the parents test, the algorithm needs a set
of generators of $\Aut(B)$ for the next step (finding the child
codes).
Description of some very effective invariants and the process of their applications are described in details in \cite{bouyukliev2007code} and \cite{mckay1990nauty}.

Similar algorithms can be used to construct linear codes with a prescribed fixed part - a residual code or a subcode.

\section{Extending linear codes via lattice point enumeration}
\label{sec_extending}

As mentioned in the introduction, we represent an $[n,k]_q$ code by a systematic generator matrix $G\in\mathbb{F}_q^{k\times n}$, i.e.,
$G$ is of the form $G=\left(I_k|R\right)$, where $I_k$ is the $k\times k$ unit matrix and $R\in\mathbb{F}_q^{k\times (n-k)}$. While this
representation is quite compact, it nevertheless can cause serious storage requirements if the number of codes get large.

Our general strategy to enumerate linear codes is to start from a (systematic) generator matrix $G$ of a code and to extend $G$ to a
generator matrix $G'$ of a {\lq\lq}larger{\rq\rq} code. Of course, there are several choices how the shapes of the matrices $G$ and $G'$
can be chosen, see e.g.\ \cite{ostergaard2002classifying} for some variants. Here we assume the form
$$
  G'=\begin{pmatrix}
    I_k & 0\dots 0  & R \\
    0   & \underset{r}{\underbrace{1\dots 1}} & \star
  \end{pmatrix}
$$
where $G=\left(I_k|R\right)$ and $r\ge 1$. Note that if $G$ is a systematic generator matrix of an $[n,k]_q$ code, then $G'$ is a
systematic generator matrix of an $[n+r,k+1]_q$ code. Typically there will be several choices for the $\star$s and some of these can
lead to equivalent codes. So, in any case we will have to face the problem that we are given a set $\mathcal{C}$ of linear codes and
we have 	to sift out all equivalent copies. 
A classical approach for this problem is to reformulate the linear code as
a graph, see \cite{bouyukliev2007code}, and then to compare canonical forms of graphs using the software package \texttt{Nauty}
\cite{mckay1990nauty}, see also \cite{ostergaard2002classifying}. In our software we use the implementation from \texttt{Q-Extension}
as well as another direct algorithmic approach implemented in the software \texttt{CodeCan} \cite{feulner2009automorphism}. In our
software, we can switch between these two tools to sift out equivalent copies and we plan to implement further variants. The reason
to choose two different implementations for the same task is to independently validate results.

It remains to solve the extension problem from a given generator matrix $G$ to all possible extension candidates $G'$.  To this end we
utilize the geometric description of the linear code generated by $G$ as a multiset $\mathcal{M}$ of points in $\operatorname{PG}(k-1,\mathbb{F}_q)$,
where
$$\mathcal{M}=\left\{\left\{  \langle g^i\rangle\,:\, 1\le i\le n  \right\}\right\},\footnote{We use the notation $\{\{\cdot\}\}$ to emphasize that
  we are dealing with multisets and not ordinary sets. A more precise way to deal with a multiset $\mathcal{M}$ in $\operatorname{PG}(k-1,\mathbb{F}_q)$
  is to use a characteristic function $\chi$ which maps each point $P$ of $\operatorname{PG}(k-1,\mathbb{F}_q)$ to an integer, which is the number
  of occurences of $P$ in $\mathcal{M}$. With this, the cardinality $\#\mathcal{M}$ can be writen as the sum over $m(P)$ for all points
  $P$ of $\operatorname{PG}(k-1,\mathbb{F}_q)$.}
$$
$g^i$ are the $n$ columns of $G$, and $\langle v\rangle$ denotes the row span of a column vector $v$. In general, the $1$-dimensional subspaces of
$\mathbb{F}_q^k$ are the points of $\operatorname{PG}(k-1,\mathbb{F}_q)$. The $(k-1)$-dimensional subspaces of $\mathbb{F}_q^k$ are called the
hyperplanes of $\operatorname{PG}(k-1,\mathbb{F}_q)$. By $m(P)$ we denote the multiplicity of
a point $P\in\mathcal{M}$. We also say that a column $g^i$ of the generator matrix has multiplicity $m(P)$, where $P=\langle g^i\rangle$ is the
corresponding point, noting that the counted columns can differ by a scalar factor. Similarly, let $\mathcal{M}'$ denote the multiset of points
in $\operatorname{PG}((k+1)-1,\mathbb{F}_q)$ that corresponds to the code generated by the generator matrix $G'$. Note that our notion of equivalent
linear codes goes in line with the notion of equivalent multisets of points in projective spaces, see \cite{dodunekov1998codes}. Counting column
multiplicities indeed partially takes away the inherent symmetry of the generator matrix of a linear code, i.e., the ordering of the columns and
multiplications of columns with non-zero field elements is not specified explicitly any more. If the column multiplicity of every column is exactly one,
then the code is called projective.

Our aim is to reformulate the extension problem $G\rightarrow G'$ as an enumeration problem of integral points in a polyhedron. 
Let $W\subseteq \{i\Delta\,:a\le i\le b\}\subseteq \mathbb{N}_{\ge 1}$ be a set of feasible weights for the non-zero codewords, where we assume $1\le a\le b$
and $\Delta\ge 1$.\footnote{Choosing $\Delta=1$ such a representation is always possible. Moreover, in many applications we can choose $\Delta>1$ quite
naturally. I.e., for optimal binary linear $[n,k,d]_2$ codes with even minimum distance $d$, i.e., those with maximum possible $d$,  we can always assume
that there exists an \emph{even} code, i.e., a code where all weights are divisible by $2$.} Linear codes where all weights of the codewords are divisible
by $\Delta$ are called $\Delta$-divisible and introduced by Ward, see e.g.~\cite{ward2001divisible,ward1981divisible}.

The non-zero codewords of the code generated by the generator
matrix $G$ correspond to the non-trivial linear combinations of the rows of $G$ (over $\mathbb{F}_q$). In the geometric setting, i.e., where an $[n,k]_q$ code $C$
is represented by a multiset $\mathcal{M}$, each non-zero codeword $c\in C$ corresponds to a hyperplane $H$ of the projective space
$\operatorname{PG}(k-1,\mathbb{F}_q)$. (More precisely, $\mathbb{F}_q^*\cdot c$ is in bijection to $H$, where $\mathbb{F}_q^*=\mathbb{F}_q\backslash\{0\}$.)
With this, the Hamming weight of a codeword $c$ is given by $$n-\sum_{P\in\operatorname{PG}(k-1,\mathbb{F}_q)\,:\,P \in \mathcal{M},\,P\le H} m(P),$$ see
\cite{dodunekov1998codes}. By $\mathcal{P}_{k}$ we denote the set of points of $\operatorname{PG}(k-1,\mathbb{F}_q)$ and by $\mathcal{H}_k$ the set of hyperplanes.

\begin{lemma}
  \label{lemma_ILP}
  Let $G$ be a systematic generator matrix of an $[n,k]_q$ code $C$ whose non-zero weights are contained in $\{i\Delta\,:a\le i\le b\}\subseteq \mathbb{N}_{\ge 1}$.
  By $c(P)$ we denote the number of columns of $G$ whose row span equals $P$ for all points $P$ of $\operatorname{PG}(k-1,\mathbb{F}_q)$ and set $c(\mathbf{0})=r$
  for some integer $r\ge 1$. With this let $\mathcal{S}(G)$ be the set of feasible solutions of
  \begin{eqnarray}
    \Delta y_H+\sum_{P\in\mathcal{P}_{k+1}\,:\,P\le H} x_P =n-a\Delta&&\forall H\in\mathcal{H}_{k+1}\label{eq_hyperplane}\\
    \sum_{q\in\mathbb{F}_q} x_{\langle (u |q)\rangle } =c(\langle u\rangle ) && \forall \langle u\rangle \in\mathcal{P}_k \cup\{\mathbf{0}\} \label{eq_c_sum}\\
    x_{\langle e_i\rangle}\ge 1&&\forall 1\le i\le k+1\label{eq_systematic}\\
    x_P\in \mathbb N &&\forall P\in\mathcal{P}_{k+1}\\
    y_H\in\{0,...,b-a\} && \forall H\in\mathcal{H}_{k+1}\label{hyperplane_var},
  \end{eqnarray}
  where $e_i$ denotes the $i$th unit vector in $\mathbb{F}_q^{k+1}$. Then, for every systematic generator matrix $G'$ of an $[n+r,k+1]_q$ code $C'$
  whose first $k$ rows coincide with $G$ and whose weights of its non-zero codewords are contained in $\{i\Delta\,:\, a\le i\le b\}$, we have a solution
  $(x,y)\in\mathcal{S}(G)$ such that $G'$ has exactly $x_P$ columns whose row span is equal
  to $P$ for each $P\in\mathcal{P}_{k+1}$.
\end{lemma}
\begin{proof}
  Let such a systematic generator matrix $G'$ be given and $x_P$ denote the number of columns of $G'$ whose row span is  equal to $P$ for all points $P\in\mathcal{P}_{k+1}$.
  Since $G'$ is systematic, Equation~(\ref{eq_systematic}) is satisfied. As $G'$ arises by appending a row to $G$, also Equation~(\ref{eq_c_sum}) is satisfied
  for all $P\in \mathcal{P}_k$. For $P=\mathbf{0}$ Equation~(\ref{eq_c_sum}) is just the specification of $r$. Obviously, the $x_P$ are non-negative
  integers. The conditions (\ref{eq_hyperplane}) and (\ref{hyperplane_var}) correspond to the restriction that the weights are contained in
  $\{i\Delta\,:\, a\le i\le b\}$.
\end{proof}
We remark that some of the constraints (\ref{eq_hyperplane}) are automatically satisfied since the subcode $C$ of $C'$ satisfies all constraints on the weights. If there
are further forbidden weights in $\{i\Delta\,:a\le i\le b\}$ then, one may also use the approach of Lemma~\ref{lemma_ILP}, but has to filter out the integer solutions that
correspond to codes with forbidden weights. Another application of this first generate, then filter strategy is to remove some of the constraints (\ref{eq_hyperplane}), which
speeds up, at least some, lattice point enumeration algorithms. In our implementation we use \texttt{Solvediophant} \cite{wassermann2002attacking}, which is based on the
LLL algorithm \cite{lenstra1982factoring}, to enumerate the integral points of the polyhedron from Lemma~\ref{lemma_ILP}.

Noting that each $[n',k',W]_q$ code, where $W\subseteq \mathbb{N}$ is a set of weights, can indeed be obtained by extending\footnote{This operation is also called \emph{lengthening}
in the coding theoretic literature, i.e.,
both the effective length $n$ and the dimension $k$ is increased, while one usually assumes that the redundancy $n-k$ remains fix. The reverse operation is called
\emph{shortening}.} all possible $[n'-r,k'-1,W]_q$ codes via Lemma~\ref{lemma_ILP}, where $1\le r\le n'-k'+1$,
already gives an algorithm for enumerating and classifying $[n',k',W]_q$ codes. (For $k'=1$ there exists a unique code for each weight $w\in W$, which admits
a generator matrix consisting of $w$ ones.) However, the number of codes $C$ with generator matrix $G$ that yield the same $[n',k',W]_q$ code $C'$ with generator
matrix $G'$ can grow exponentially with $k'$. We can limit this growth a bit by studying the effect of the extension operation and its reverse on some code invariants.

\begin{lemma}
  \label{lemma_shortening}
  Let $C'$ be an $[n',k',W]_q$ code with generator matrix $G'$. If $G'$ contains a column $g'$ of multiplicity $r\ge 1$, then there exists a generator matrix $G$ of an $[n'-r,k'-1,W]_q$
  code $C$ such that the extension of $G$ via Lemma~\ref{lemma_ILP} yields at least one code that is equivalent to $C'$.
  Moreover, if $\Lambda$ is the maximum column multiplicity of $G'$, without counting the columns whose row span equals $\langle g'\rangle$, then the maximum column multiplicity of
  $G$ is at least $\Lambda$.
\end{lemma}
\begin{proof}
  Consider a transform $\tilde{G}$ of $G'$ such that the column $g'$ of $G'$ is turned into the $j$th unit vector $e_j$ for some integer $1\le j\le k'$. Of course also $\tilde{G}$ is
  a generator matrix of $C'$. Now let $\hat{G}$ be the $(k'-1)\times (n'-r)$-matrix over $\mathbb{F}_q$ that arises from $\tilde{G}$ after removing the $r$ occurrences of the
  columns with row span $\langle e_j\rangle$ and additionally removing the $j$th row. Note that the non-zero weights of the linear code generated by $\hat{G}$ are also contained
  in $W$. If $G$ is a systematic generator matrix of the the linear code $C$ generated by $\hat{G}$, then Lemma~\ref{lemma_ILP} applied to $G$ with the chosen parameter $r$ yields
  especially a linear code with generator matrix $G'$ as a solution. By construction the effective length of $C$ is indeed $n'-r$. Finally, note that removing a row from a generator
  matrix does not decrease column multiplicities.
\end{proof}

\begin{corollary}
  \label{cor_shortening}
  Let $C'$ be an $[n',k',W]_q$ code with generator matrix $G'$ and minimum column multiplicity $r$. Then there exists a generator matrix $G$ of an $[n'-r,k'-1,W]_q$
  code $C$ with minimum column multiplicity at least $r$ such that the extension of $G$ via Lemma~\ref{lemma_ILP} yields at least one code that is equivalent $C'$.
\end{corollary}

Corollary~\ref{cor_shortening} has multiple algorithmic implications. If we want to classify all $[n,k,W]_q$ codes, then we need the complete lists of $[\le n-1,k-1,W]_q$ codes,
where $[\le n',k',W'_q]$ codes are those with an effective length of at most $n'$. Given an $[n',k-1,W]_q$ code with $n'\le n-1$ we only need to extend those codes which
have a minimum column multiplicity of at least $n-n'$ via Lemma~\ref{lemma_ILP}. If $n-n'>1$ this usually reduces the list of codes, where an extensions needs to be computed.
Once the set $\mathcal{S}(G)$ of feasible solutions is given, we can also sift out some solutions before applying the equivalence sifting step. Corollary~\ref{cor_shortening}
allows us to ignore all resulting codes which have a minimum column multiplicity strictly smaller than $n-n'$. Note that when we know $x_{P}>0$, which we do know e.g.\
for $P=\langle e_i\rangle$, where $1\le i\le k+1$, then we can add the valid inequality $x_P\ge n-n'$ to the inequality system from Lemma~\ref{lemma_ILP}. We call the application
of the extension step of Lemma~\ref{lemma_ILP} under these extra assumptions \emph{canonical length extension} or \emph{canonical lengthening}.

As an example we consider the $[7,2]_2$ code that arises from two codewords of Hamming weight $4$ whose support intersect in cardinality $1$, i.e., their sum has Hamming weight $6$.
A direct construction gives the generator matrix
$$
 G_1=
  \begin{pmatrix}
    1 & 1 & 1 & 1 & 0 & 0 & 0\\
    0 & 0 & 0 & 1 & 1 & 1 & 1
  \end{pmatrix},
$$
which can be transformed into
$$
 G_2=
  \begin{pmatrix}
    1 & 1 & 1 & 1 & 0 & 0 & 0\\
    1 & 1 & 1 & 0 & 1 & 1 & 1
  \end{pmatrix}.
$$
Now column permutations are necessary to obtain a systematic generator matrix
$$
 G_3=
  \begin{pmatrix}
    1 & 0 & 0 & 0 & 1 & 1 & 1 \\
    0 & 1 & 1 & 1 & 1 & 1 & 1
  \end{pmatrix}.
$$
Note that $G_2$ and $G_3$ do not generate the same but only equivalent codes. Using the canonical length extension the systematic
generator matrix
$$
 G_0=
  \begin{pmatrix}
    1 & 1 & 1 & 1
  \end{pmatrix}
$$
of a single codeword of Hamming weight $4$ cannot be extended to $G_3$, since we would need to choose $r=3$ to get from a $[4,1]_2$ code to a
$[7,2]_2$ code, while the latter code has a minimum column multiplicity of $1$. However, the unique codeword with Hamming weight $6$ and
systematic generator matrix
$$
  G=
  \begin{pmatrix}
    1 & 1 & 1 & 1 & 1 & 1
  \end{pmatrix}
$$
can be extended to
$$
 G_4=
  \begin{pmatrix}
    1 & 0 & 1 & 1 & 1 & 1 & 1 \\
    0 & 1 & 0 & 0 & 1 & 1 & 1
  \end{pmatrix},
$$
which generates the same code as $G_3$. So, we needed to consider an extension of a $[6,1]_2$ code to a $[7,2]_2$ code. Now let us dive into the
details of the integer linear programming formulation of Lemma~\ref{lemma_ILP}. In our example we have $k=1$ and $q=2$, so that
$\mathcal{P}_1=\left\{\langle (1)\rangle\right\}$, and
$$
  \mathcal{P}_2=\left\{
  \left\langle\begin{pmatrix}1\\0\end{pmatrix}\right\rangle,\left\langle\begin{pmatrix}0\\1\end{pmatrix}\right\rangle,\left\langle\begin{pmatrix}1\\1\end{pmatrix}\right\rangle
  \right\}.
$$
The multiplicities corresponding to the columns of $G$ and $r$ are given by
$$
  c(\langle(1)\rangle)=6\quad\text{and}\quad c(\langle(0)\rangle)=1.
$$
Due to constraint~(\ref{eq_c_sum}) we have
$$
  x_{\langle e_1\rangle}+x_{\langle e_1+e_2\rangle} =6\quad\text{and}\quad x_{\langle e_2\rangle}=1.
$$
Constraint~(\ref{eq_systematic}) reads
$$
  x_{\langle e_1\rangle}\ge 1\quad\text{and}\quad x_{\langle e_2\rangle}\ge 1.
$$
In order to write down constraint~(\ref{eq_hyperplane}), we need to specify the set $W$ of allowed weights. Let us choose $W=\{4,6\}$, i.e., $\Delta=2$, $a=2$, and $b=3$.
If we label the hyperplanes by $\mathcal{H}=\left\{1,2,3\right\}$, for the ease of notation, we obtain
\begin{eqnarray*}
  2y_1+x_{\langle e_2\rangle} &=& 3,\\
  2y_2+x_{\langle e_1+e_2\rangle} &=& 3,\text{ and}\\
  2y_3+x_{\langle e_1\rangle} &=& 3.
\end{eqnarray*}
Since the $y_i$ are in $\{0,1\}$ we have $x_{\langle e_1\rangle}\le 3$ and $x_{\langle e_1+e_2\rangle}\le 3$, so that
$x_{\langle e_1\rangle}= 3$ and $x_{\langle e_1+e_2\rangle}=3$. The remaining variables are given by $x_{\langle e_2\rangle}=1$, $y_1=1$, $y_2=0$, and $y_3$.
Thus, in our example there is only one unique solution, which then corresponds to generator matrix $G_4$ (without specifying the exact ordering of the columns
of $G_4$).

Note that for the special situation $k+1=2$, every hyperplane of $\mathcal{P}_2$ consists of a unique point. The set of column or point multiplicities is left invariant
by every isometry of a linear code. For hyperplanes in $\operatorname{PG}(k+1,\mathbb{F}_q)$ or non-zero codewords of $C'$ a similar statement applies. To this end we
introduce the weight enumerator $w_C(x)=\sum_{i=0}^n A_ix^i$ of a linear code $C$, where $A_i$ counts the number of codewords of Hamming weight exactly $i$ in $C$. Of
course, the weight enumerator $w_C(x)$ of a linear code $C$ does not depend on the chosen generator matrix $C$. The geometric reformulation uses the number $a_i$
of hyperplanes $H\in\mathcal{H}_k$ with $\# H\cap \mathcal{M}:=\sum_{P\in\mathcal{P}_k\,:\,P\in \mathcal{M},\,P\le H} m(P)=i$. The counting vector $\left(a_0,\dots, a_n\right)$
is left unchanged by isometries. One application of the weight enumerator in our context arises when we want to sift out equivalent copies from a list $\mathcal{C}$ of
linear codes. Clearly, two codes whose weight enumerators do not coincide, cannot be equivalent. So, we can first split $\mathcal{C}$ according to the occurring
different weight enumerators and then apply one of the mentioned algorithms for the equivalence filtering on the smaller parts separately. We can even refine this
invariant a bit more. For a given $[n,k]_q$ code $C$ with generator matrix $G$ and corresponding multiset $\mathcal{M}$ let $\widetilde{\mathcal{M}}$ be the set of
different elements in $\mathcal{M}$, i.e., $\#\mathcal{M}=\sum_{P\in\widetilde{\mathcal{M}}} m(P)$, which means that we ignore the multiplicities in
$\widetilde{\mathcal{M}}$. With this we can refine Lemma~\ref{lemma_shortening}:

\begin{lemma}
  \label{lemma_shortening_refined}
  Let $C$ be an $[n,k,W]_q$ code with generator matrix $G$ and $\mathcal{M}$, $\widetilde{\mathcal{M}}$ as defined above. For each $P\in \widetilde{\mathcal{M}}$
  there exists a generator matrix $G_P$ of an $[n-m(P),k-1]_q$ code such that the extension of $G_P$ via Lemma~\ref{lemma_ILP} yields at least one code that is equivalent to $C$.
\end{lemma}

Now we can use the possibly different weight enumerators of the subcodes generated by $G_P$ to distinguish some of the extension paths.

\begin{corollary}
  \label{cor_shortening_refined}
  Let $C'$ be an $[n',k',W]_q$ code with generator matrix $G'$, minimum column multiplicity $r$, and $\mathcal{M}$, $\widetilde{\mathcal{M}}$ as defined above. Then there exists a
  generator matrix $G$ of an $[n'-r,k'-1,W]_q$ code $C$ such that the extension of $G$ via Lemma~\ref{lemma_ILP} yields at least one code that is equivalent to $C'$ and
  the weight enumerator $w_C(x)$ is lexicographically minimal among the weight enumerators $w_{C_P}(x)$ for all $P\in\widetilde{\mathcal{M}}$ with column multiplicity $r$ in $C'$,
  where $C_P$ is the linear code generated by the generator matrix $G_P$ from Lemma~\ref{lemma_shortening_refined}.
\end{corollary}

We remark that the construction for subcodes, as described in Lemma~\ref{lemma_shortening_refined}, can also be applied for points
$P\in\mathcal{P}_k\backslash\mathcal{M}$. And indeed, we obtain an $[n-m(P),k-1]_q=[n,k-1]_q$ code, i.e., the effective length does not decrease, while the dimension decreases
by one.

The algorithmic implication of Corollary~\ref{cor_shortening_refined} is the following. Assume that we want to extend an $[n,k,W]_q$ code $C$ with generator matrix
$G$ to an $[n+r,k+1,W]_q$ code $C'$ with generator matrix $G'$. If the minimum column multiplicity of $C$ is strictly smaller than $r$, then we do not need to
compute any extension at all. Otherwise, we compute the set $\mathcal{S}(G)$ of solutions according to Lemma~\ref{lemma_ILP}. If a code $C'$ with generator matrix
$G'$, corresponding to a solution in $\mathcal{S}(G)$, has a minimum column multiplicity which does not equal $r$, then we can skip this specific solution. For all other
candidates let $\overline{\mathcal{M}}\subseteq \mathcal{P}_{k+1}$ the set of all different points spanned by the columns of $G'$ that have multiplicity exactly $r$.
By our previous assumption $\overline{\mathcal{M}}$ is not the empty set. If $w_C(x)$ is the lexicographically minimal weight enumerator among all weight enumerators $w_{C_P}(x)$,
where $P\in\overline{\mathcal{M}}$ and $C_P$ is generated by the generator matrix $G_P$ from Lemma~\ref{lemma_shortening_refined}, then we store $C'$ and skip it
otherwise. We call the application of the extension step of Lemma~\ref{lemma_ILP} under these extra assumptions \emph{lexicographical extension} or \emph{lexicographical lengthening}.

Lexicographical lengthening drastically decrease the ratio between the candidates of linear codes that have to be sifted out and the resulting number of non-equivalent codes.
This approach also allows parallelization of our enumeration algorithm, i.e., given an exhaustive list $\mathcal{C}$ of all $[n,k,W]_q$ codes and an integer $r\ge 1$, we can
split $\mathcal{C}$ into subsets $\mathcal{C}_1,\dots,\mathcal{C}_l$ according to their weight enumerators. If the $[n+r,k+1,W]_q$ code $C'$ arises by lexicographical lengthening
from a code in $\mathcal{C}_i$ and the $[n+r,k+1,W]_q$ code $C''$ arises by lexicographical lengthening from a code in $\mathcal{C}_j$, where $i\neq j$, then $C'$ and $C''$
cannot be equivalent. As an example, when constructing the even $[21,8,6]_2$ codes from the $17\,927\,353$ $[20,7,6]_2$ codes, we can split the construction into more than
$1000$ parallel jobs. If we do not need the resulting list of $1\,656\,768\,624$ linear codes for any further computations, there is no need to store the complete list of codes
during the computation.

\section{Numerical results}
\label{sec_results}

We implemented the presented algorithms in the programs \texttt{Generation} and \texttt{LinCode}.
These algorithms can be used to classify linear codes with wide-range parameters, for example, for binary codes with lengths up to 100 or more depending on the dimension. The main objectives in this section are three: (1) to show what problems in Coding Theory can be attacked with the software presented, (2) to solve given classification problems, and (3) to show what these results can be useful for (Subsection~\ref{subsec_applications}). The presented classification results are of three types: (1) binary codes with prescribed minimum distance; (2) divisible codes over fields with 2, 3 and 4 elements; (3) self-orthogonal codes.

As the implementation of a practically efficient algorithm for the classification of linear codes is a delicate issue, we exemplarily verify
several classification results from the literature. Efficiency is demonstrated by partially extending some of these enumeration results.

\subsection{Results.}

In \cite[Research Problem 7.2]{kaski2006classification} the authors ask for the classification of $[n,k,3]_2$ codes for $n>14$. In Table~\ref{tab_n_k_3_16}
we extend their Table 7.7 to $n\le 18$.

\begin{table}[htp]
  \begin{center}
  {\footnotesize
    \begin{tabular}{r|rrrrrrrrrr}
      \hline
      $n/k$ & 2 & 3 & 4 & 5 & 6 & 7 & 8 & 9 & 10\\
      \hline
       5 & 1 \\
       6 & 3 & 1 \\									
       7 & 4 & 4 & 1 \\								
       8 & 6 & 10 & 5 \\								
       9 & 8 & 23 & 23 & 5 \\							
      10 & 10 & 42 & 76 & 41 & 4 \\						
      11 & 12 & 71 & 207 & 227 & 60 & 3 \\					
      12 & 15 & 115 & 509 & 1012 & 636 & 86 & 2 \\				
      13 & 17 & 174 & 1127 & 3813 & 4932 & 1705 & 110 & 1 \\			
      14 & 20 & 255 & 2340 & 12836 & 31559 & 24998 & 4467 & 127 & 1 \\		
      15 & 23 & 364 & 4606 & 39750 & 176582 & 293871 & 132914 & 11507 & 143 \\	
      16 & 26 & 505 & 8685 & 115281 & 896316 & 2955644 & 3048590 & 733778 & 28947 \\
      17 & 29 & 686 &15797 & 317464 &4226887 & 26590999 & 58085499 & 34053980 &4115973 \\
      18 & 33 & 919 &27907 & 837697 &18807438& 220135857& 971007974& 1261661451& 393087258\\
      \hline
      \hline
      $n/k$ & 11 & 12 & 13\\
      \hline
      15 &  1 \\	
      16 &  144\\
      17 & 70455 & 129 \\
      18 & 27333440 &293458 & 226\\
      \hline
    \end{tabular}}
    \smallskip
    \caption{The number of inequivalent $[n,k,3]_2$ codes for $n\le 18$}
    \label{tab_n_k_3_16}
  \end{center}
\end{table}

Blank entries correspond to the non-existence of any code with
these parameters, i.e., there is no $[4,2,3]_2$ code and also no $[16,12,3]_2$ code. Obviously, there is a unique $[n,1,3]_2$ codes for each $n\ge 3$ and
it is not too hard to show that the number of inequivalent $[n,2,3]_2$ codes is given by
$
  \left\lceil\sqrt{\frac{(n-4)(n-3)(2n-7)}{6}}\,\right\rceil
$
for each $n\ge 3$. 

\begin{table}[htp]
  \begin{center}
    \begin{tabular}{c|ccccccccccc}
      \hline
      $k$ &
      4 & 5 & 6 & 7 & 8 & 9 & 10 & 11 & 12 & 13 \\
      \hline
      \# &
      8561\!&\!129586\!&\!1813958\!&\!16021319\!&\!60803805\!&\!73340021\!&\!22198835\!&\!1314705\!&\!11341\!&\!24 \\
      \hline
    \end{tabular}
    \smallskip
    \caption{The number of inequivalent even $[n\le 19,k,4]_2$ codes for $4\le k\le 13$}
    \label{tab_n_k_4_even}
  \end{center}
\end{table}

There are possibilities for different restrictions for the codes in addition to the restrictions on length, dimension, minimum and dual distances. We apply also restrictions on the orthogonality and weights of the codewords in some examples.

We present the counts for the even $[\le 19,k,4]_2$ codes in Table~\ref{tab_n_k_4_even}. The numbers of the inequivalent even $[n\le 21,k,6]_2$ codes are presented in
Table~\ref{tab_n_k_6_even} (excluding the
enumeration of the even $[21,9,6]_2$ codes because their number is extremely huge).
We have verified these results by both software programs.


\begin{table}[htp]
  \begin{center}
    \begin{tabular}{c|ccccccccccc}
      \hline
      $k$ & 3 & 4 & 5 & 6 & 7 & 8 & 10 & 11  \\
      \hline
      \# & 726 & 12817 & 358997 & 11697757 & 246537467 & 1697180017 & 62180809 & 738 \\
      \hline
    \end{tabular}
    \smallskip
    \caption{The number of even $[\le 21,k,6]_2$ codes for $3\le k\le 11$, $k\neq 9$}
    \label{tab_n_k_6_even}
  \end{center}
\end{table}

\begin{table}[htp]
  \begin{center}
{\footnotesize
\begin{tabular}{c|ccccc}
\hline\noalign{\smallskip}
$n\setminus k$ & 2 & 3&4&5&6  \\
\noalign{\smallskip}\hline\noalign{\smallskip}
12 & 1 &    &     &     &\\
13 &   &  1 &     &     &\\
18 & 1 &    &     &     &\\
21 & 1 &  1 &     &     &\\
22 &   &  1 &   1 &     &\\
24 & 1 &  1 &   1 &     &\\
25 &   &  1 &   1 &   1 &\\
26 &   &  1 &   1 &   1 &   1\\
27 & 2 &  3 &   3 &   1 &\\
30 & 2 &  4 &   3 &     &\\
31 &   &  2 &   3 &   1 &\\
33 & 1 &  5 &   5 &   3 &\\
34 &   &  2 &   5 &   4 &   1\\
35 &   &  1 &   4 &   4 &   3\\
36 & 4 & 10 &  22 &  13 &   4\\
37 &   &  2 &   7 &  10 &   3\\
38 &   &  1 &   6 &  12 &  10\\
39 & 3 & 15 &  34 &  41 &  23\\
40 &   &  6 &  25 &  40 &  30\\
41 & 0 &  0 &  0 &  0 &  0\\
42 & 2 & 17 &  52 &  44 &  15\\
43 &   &  6 &  32 &  40 &  16\\
44 &   &  2 &  14 &  22 &  17\\
45 & 5 & 31 & 141 & 190 &  72\\
46 &   &  6 &  56 & 122 &  71\\
47 &   &  2 &  29 &  92 &  89\\
48 & 5 & 44 & 297 & 705 & 468\\
49 &   & 15 & 177 & 613 & 596\\
50 &   &  2 &  39 & 217 & 295\\
\noalign{\smallskip}\hline\noalign{\smallskip}
total & 28 & 182&958&2176&1714 \\                                                                                                                                                                                                                                                                                                                                                                                                                                                                                                                                                \noalign{\smallskip}\hline
\end{tabular}}
\smallskip
\caption{Divisible ternary codes with $n\le 50$, $k\le 6$, $\Delta=9$}
\label{tab_n_k_9_div_q_3}       
\end{center}
\end{table}

Moreover, we have enumerated the divisible codes with given parameters over fields with 2, 3 and 4 elements. Recently, 8-divisible binary codes (called also triply even) have been investigated \cite{triply-Munemasa,no59}.
In \cite{ubt_eref40887}, it is proven that projective triply-even binary
codes exist precisely for lengths 15, 16, 30, 31, 32, $45-51$, and $\ge 60$. We have verified the enumeration of the projective $2$, $4$-, and $8$-divisible binary linear codes from \cite{ubt_eref40887}.

We have classified 9-divisible ternary codes and 4-divisible quaternary codes.

\begin{itemize}
\item $q=3$, $\Delta=9$. Table~\ref{tab_n_k_9_div_q_3} contains classification results for codes of this type with length $n\le 55$ and dimension $k\le 8$. The conspicuous zero row
for length $n=41$ has a theoretical explanation, i.e., there is no $9$-divisible $[41,k]_3$ code at all, see \cite[Theorem 1]{divisibleIEEE}.\footnote{More precisely,
$41=2\cdot 13+2\cdot 12-1\cdot 9$ is a certificate for the fact that such a code does not exist, see \cite[Theorem 1, Example 6]{divisibleIEEE}.}


\item $q=4$, $\Delta=4$. Table \ref{table-q4-n30} presents classification results for codes with $n\le 30$ and $k\le 8$. All constructed codes are Hermitian self-orthogonal.
\end{itemize}

\begin{table}[htp]
  \begin{center}
{\footnotesize
\begin{tabular}{c|ccccccc}
\hline\noalign{\smallskip}
$n\setminus k$ & 2 & 3&4&5&6&7&8  \\
\noalign{\smallskip}\hline\noalign{\smallskip}
5  & 1 &     &      &     &&&\\
8  & 1 &     &      &     &&&\\
9  & 1 &   1 &      &       &    &&\\
10 & 1 &   1 &    1 &       &    &&\\
12 & 2 &   2 &      &       &    &&\\
13 & 2 &   3 &    1 &       &    &&\\
14 & 1 &   5 &    3 &     1 &    &&\\
15 & 1 &   3 &    6 &     2 &     1 &&\\
16 & 4 &   9 &    7 &     2 &       &&\\
17 & 3 &  12 &    9 &     2 &       &&\\
18 & 2 &  18 &   25 &     8 &     1 &&\\
19 & 1 &  14 &   42 &    25 &     6 &     1 &\\
20 & 6 &  34 &   93 &    70 &    22 &     4 &1\\
21 & 5 &  45 &  115 &    75 &    19 &     2 &\\
22 & 3 &  64 &  245 &   131 &    23 &     2 &\\
23 & 2 &  62 &  554 &   398 &    96 &    12 &     1\\
24 & 9 & 123 & 1509 &  1769 &   491 &    79 &     9\\
25 &   & 168 & 3189 &  6890 &  1842 &   334 &    46\\
26 &   &     & 8420 & 18377 &  2691 &   360 &    33\\
27 &   &     &      & 70147 &  4602 &   458 &    34\\
28 &   &     &      &       & 36982 &  3075 &   244\\
29 &   &     &      &       &       & 34180 &  2366\\
30 &   &     &      &       &       &       & 24565\\                                                                                                                                                                                                                                                                                                                                                                                                                                                                                                                            \noalign{\smallskip}\hline\noalign{\smallskip}
total &  45& 564  &14219  & 97897 & 46776 & 38507 & 27299  \\                                                                                                                                                                                                                                                                                                                                                                                                                                                                                                                                                \noalign{\smallskip}\hline
\end{tabular}}
\smallskip
\caption{Divisible quaternary codes with $n\le 30$, $k\le 8$, $\Delta=4$}
\label{table-q4-n30}
\end{center}
\end{table}

 By the program \texttt{Generation}, we have also classified binary, ternary and quaternary self-orthogonal codes. There are a few tables of self-orthogonal codes (see \cite{BBGO,IliyaPatric}). Here we present classification results that are not given in these tables, namely:
\begin{itemize}
\item We present classification results for binary self-orthogonal $[27,k\le 12,d\ge 8]$ codes with dual distance $d^\perp\ge 1$ in Table \ref{table-q2-n27}. The codes with dimensions 11 and 12 are optimal as linear codes, and the codes with $k=9$ and 10 are optimal only as self-orthogonal \cite{BBGO}. Moreover, we tried to fill some of the gaps in \cite[Table 1]{BBGO}. We classified the $n$-optimal self-orthogonal $[n,k,d]$ codes (the codes for which no $[n-1,k,d]$ self-orthogonal code exists) with parameters $[35,8,14]$, $[29,9,10]$ and $[30,10,10]$. The number of codes in these cases are 376, 36504 and 573, respectively. Our program shows that no self-orthogonal $[37,10,14]$ and $[36,9,14]$ codes exist which means that the maximal possible minimum distance for self-orthogonal codes with these lengths and dimensions is 12.

\begin{table}[htp]
\begin{center}
\begin{tabular}{c|cccccc}
\hline\noalign{\smallskip}
k & 2 & 3&4&5&6&7  \\
\noalign{\smallskip}\hline\noalign{\smallskip}
total & 59  & 445 &4615  & 64715  & 959533&8514764 \\
\noalign{\smallskip}\hline\noalign{\smallskip}
k & 8 & 9&10&11&12&  \\
\noalign{\smallskip}\hline\noalign{\smallskip}
total &  21256761 &7030920  &159814  &791   &18 & \\
\noalign{\smallskip}\hline
\end{tabular}
\smallskip
\caption{Binary self-orthogonal $[27,k\le 12,d \ge 8]d^\perp\ge 1$ codes}
\label{table-q2-n27}       
\end{center}
\end{table}


\item  The classification results for ternary self-orthogonal $[n\le 20,k\le 10,d\ge 6]$ codes are given in Table \ref{table-q3-n20all}. This table supplements \cite[Table 1]{IliyaPatric}.


\begin{table}[htp]
\begin{center}
\begin{tabular}{c|ccccccc}
\hline\noalign{\smallskip}
$n\setminus k$ &4&5&      6&       7&     8&   9&10\\
\noalign{\smallskip}\hline\noalign{\smallskip}
10   &1   &       &        &        &      &    &\\
11   &1   &      1&        &        &      &    & \\
12   &6   &      2&       1&        &      &    & \\
13   &10  &      4&       1&        &      &    & \\
14   &27  &     15&       4&        &      &    & \\
15   &78  &     73&      20&       2&      &    &\\
16   &181 &    312&     121&      11&     1&    & \\
17   &414 &   1466&     885&      86&     2&    & \\
18   &1097&   8103&   10808&    1401&    40&    &\\
19   &2589&  47015&  167786&   45950&  1132&  10& \\
20   &6484& 285428& 2851808& 2121360& 89670& 464& 6 \\
\noalign{\smallskip}\hline
\end{tabular}
\smallskip
\caption{Ternary self-orthogonal codes with $n\le 20$, $k\le 10$, and $d\ge 6$}
\label{table-q3-n20all}       
\end{center}
\end{table}

\item Table \ref{table-q4-n21} shows the classification of the $[n\le 21,\le 6,12]$ quaternary Hermitian self-orthogonal codes. These results fill some of the gaps in \cite[Table 2]{IliyaPatric}.
\end{itemize}

\begin{table}[htp]
\begin{center}
\begin{tabular}{c|ccccc}
\hline\noalign{\smallskip}
$n\setminus k$ & 2 & 3&4&5&6  \\
\noalign{\smallskip}\hline\noalign{\smallskip}
15&                   1&&&\\
16&                   2&         1&&&\\
17&                   3&         4&         1&&\\
18&                    &        45&        12&&\\
19&                    &          &      5673&&\\
20&                    &          &          &  886576&\\
21&                    &         &&                    &   577008\\                                                                                                                                                                                                                                                                                                                                                                                                                                                                                                                                                  \noalign{\smallskip}\hline
\end{tabular}
\smallskip
\caption{Quaternary Hermitian self-orthogonal codes with $n\le 21$, $k\le 6$, $d=12$}
\label{table-q4-n21}       
\end{center}
\end{table}

\subsection{Applications}
\label{subsec_applications}

In this subsection we want to exemplarily show up, that exhaustive enumeration results of
linear codes can of course be used to obtain results for special subclasses of codes and
their properties by simply checking all codes. For our first example we remark that the
support of a codeword is the set of its non-zero coordinates. A non-zero codeword $c$ of a
linear code $C$ is called minimal if the support of no other non-zero codeword is contained
in the support of $c$, see e.g.~\cite{ashikhmin1998minimal}. By $m_2(n,k)$ we denote the minimum
number of minimal codewords of a projective\footnote{Duplicating columns in a binary
linear code generated by the $k\times k$ unit matrix results in exactly $k$ minimal codewords, which
is the minimum for all $k$-dimensional codes.} $[n,k]_2$ code. In Table~\ref{tab_minimal_codewords}
we state the exact values of $m_2(n,k)$ for all $2\le k\le n\le 15$ obtained by enumerating all
projective codes with these parameters.

\begin{table}[htbp]
\begin{center}
{\small
\begin{tabular}{|c|c|c|c|c|c|c|c|c|c|c|c|c|c|c|}\hline
$n/k$ & 2 & 3   & 4     & 5     & 6     & 7     & 8     & 9  & 10 & 11 & 12 & 13 & 14 & 15 \\\hline
3     & 3 & 3   &       &       &       &       &       &    &    &    &    &    &    & \\\hline
4     &   & 4   & 4     &       &       &       &       &    &    &    &    &    &    & \\\hline
5     &   & 6   & 5     & 5     &       &       &       &    &    &    &    &    &    & \\\hline
6     &   & 7   & 6     & 6     & 6     &       &       &    &    &    &    &    &    & \\\hline
7     &   & 7   & 8     & 7     & 7     & 7     &       &    &    &    &    &    &    & \\\hline
8     &   &     & 8     & 9     & 8     & 8     & 8     &    &    &    &    &    &    & \\\hline
9     &   &     & 12    & 9     & 9     & 9     & 9     & 9 &    &    &    &    &    & \\\hline
10    &   &     & 14    & 10    & 10    & 10    & 10    & 10 & 10 &    &    &    &    & \\\hline
11    &   &     & 14    & 15    & 11    & 11    & 11    & 11 & 11 & 11 &    &    &    & \\\hline
12    &   &     & 15    & 15    & 13    & 12    & 12    & 12 & 12 & 12 & 12 &    &    & \\\hline
13    &   &     & 15    & 16    & 14    & 13    & 13    & 13 & 13 & 13 & 13 & 13 &    & \\\hline
14    &   &     & 15    & 16    & 14    & 15    & 14    & 14 & 14 & 14 & 14 & 14 & 14 & \\\hline
15    &   &     & 15    & 16    & 17    & 15    & 16    & 15 & 15 & 15 & 15 & 15 & 15 & 15\\
\hline\end{tabular}}
\smallskip
\caption{$m_2(n,k)$ for $3\leq n\leq 15, 1\leq k\leq 9$}
\label{tab_minimal_codewords}
\end{center}
\end{table}

\subsection{Verification and computational time.}

We use three basic approaches to verify our programs and the results. The first one is verification by replication. We ran the programs to get already published classification results as the classification of the projective $2$, $4$-, and $8$-divisible binary linear codes from \cite{ubt_eref40887}, and binary projective codes with dimension 6 \cite{bouyukliev2006}, and we obtained the same number of codes.

The second approach is to double check most of the results with the presented programs. The third one is
to use enumeration of different types of codes given by theoretical methods (see \cite{betten2006error,HuffmanPless}). For self-orthogonal codes we can also use mass formulae to verify that the constructed codes represents all equivalence classes of the given length \cite{HuffmanPless}.

All calculations with the program \texttt{Generation} have been done on 2 $\times$ Intel Xeon E5-2620 V4, 32 thread processor. The calculation time strictly depends on the parameters and restrictions of the considered codes that determine the size of the corresponding search tree. For example, the proof that the extended binary Golay code is unique took 0.09s on a single tread. The classification of the binary $[32,16,8]$ self-dual codes took 184.49s also on a single tread.

As already mentioned, the presented algorithms are suitable for parallel implementation. The results in Table \ref{tab_n_k_3_16} were obtained at once in about 3 days on 32 threads. The calculations for Table \ref{tab_n_k_6_even} took about the same time on the same server. All other calculations took from a few minutes to several days on a single core.

For the program \texttt{LinCode} the most time expensive step that was executed a single computer, i.e., extending the $[19,7,6]_2$ codes to $[20,8,6]_2$ codes, took roughly
250~hours of computation time on a single core of a Intel Core i7-7600U laptop with 2.80GHz . We remark that the $[19,k,4]_2$ codes, where $k\in\{7,8,9,10\}$, and
the $[21,k,6]_2$ codes, where $k\in\{7,8,10\}$, were enumerated in parallel, i.e., we have partially used the computing nodes of the
\emph{High Performance Computing Keylab} from the University of Bayreuth. We have used the oldest cluster btrzx5 that went into operation in
2009.\footnote{The precise technical details can be found at \url{https://www.bzhpc.uni-bayreuth.de/de/keylab/Cluster/btrzx5_page/index.html}.}

\section{Conclusion}
\label{sec_conclusion}

The technique of canonical augmentation is used for classification of special types of codes and related combinatorial objects in \cite{bouyukliev2006,BB38,vanEupenLizonek1997,royle1998orderly}, etc. We apply this technique to classify linear codes with given properties in the program \texttt{Generation}.
Our algorithm expands the matrices column by column but starts from the identity $k\times k$ matrix. So it constructs all inequivalent linear $[n,k]_q$ codes without getting codes of smaller dimensions. Restrictions on the dual distance, minimum distance, etc. can be applied.
The algorithm is implemented in the program \texttt{Generation}, which is the first module of the software package \texttt{QextNewEdition}. On the one hand, this program gives us the possibility to classify linear codes with given parameters over fields with $q$ elements. On the other hand, the program can give families of inequivalent codes with certain properties that can be used for extension in length and dimension from the other modules in the package. These modules are also based on the idea of canonical augmentation, which gives the possibility for parallelization.

Moreover, we have presented an algorithm for the classification of linear codes over finite fields based on lattice
point enumeration. The lattice point enumeration itself and sifting out equivalent copies is so far done
with available scientific software packages. Using invariants like the weight enumerator of subcodes, see
Corollary~\ref{cor_shortening_refined}, the number of candidates before sifting could kept reasonably small. The advantage of the canonical augmentation that no pairs of codes have to be checked
whether they are equivalent comes at the cost that the computation of the canonical form is relatively costly. Allowing not only a single canonical extension, but a relatively small
number of extensions that may lead to equivalent codes, might be a practically efficient alternative. We have
also demonstrated that the algorithm can be run in parallel.
However, we think that our implementation can still be further improved. In some cases the used lattice point
enumeration algorithm \texttt{Solvediophant} takes quite long to verify that a certain code does not allow an
extension, while integer linear programming solvers like e.g.~\texttt{Cplex} quickly verify infeasibility.
We propose the extension of Table~\ref{tab_n_k_9_div_q_3}
as a specific open problem. We have
demonstrated that it is indeed possible to exhaustively classify sets of linear codes of magnitude $10^9$, which was
not foreseeable  at the time of \cite{kaski2006classification}.

Currently the implementation of the evolving software package \texttt{LinCode} is not advanced enough to be
made publicly available. So, we would like to ask the readers to sent their interesting enumeration problems
of linear codes directly to the third author. The software package \texttt{QextNewEdition} is available on the web-page\\
\verb"http://www.moi.math.bas.bg/moiuser/~data/Software/QextNewEdition"



\end{document}